\documentclass[a4paper]{article}
\usepackage{amsmath,amsthm,amssymb,amsfonts}
\usepackage{algorithm}
\usepackage[noend]{algpseudocode}
\usepackage{times}
\usepackage{hyperref}
\usepackage{xspace}
\usepackage{caption}
\usepackage{subcaption}
\usepackage[square,numbers,sort]{natbib}
\usepackage{sistyle}
\usepackage{graphicx}

\SIthousandsep{,}

\def\R{\mathcal{R}} 
\def\W{\mathcal{W}} 
\def\M{\mathcal{M}} 
\def\C{\mathcal{C}} 
\def\B{\mathcal{B}} 
\def\F{\mathcal{F}} 
\def\X{\mathcal{X}} 
\def\G{\mathcal{G}}
\def\H{\mathcal{H}}

\def\dR{\mathbb{R}} 

\def\eps{\varepsilon}
\def\limn{\lim_{n\rightarrow \infty}}
\def\frechet{Fr\'echet\xspace}

\newcommand{\Cpp}{C\raise.08ex\hbox{\tt ++}\xspace}

\def\np{\textsc{np}}

\def\as{a.s.\xspace}

\def\Im{\textup{Im}}

\newtheorem{lemma}{Lemma}
\newtheorem{theorem}{Theorem}

\theoremstyle{definition}
\newtheorem{definition}{Definition}

\theoremstyle{plain}

\title{Sampling-based bottleneck pathfinding with applications to
  Fr\'echet matching
  \thanks{This work has been supported in part by the Israel Science
    Foundation (grant no.~1102/11), by the Blavatnik Computer Science
    Research Fund, and by the Hermann Minkowski--Minerva Center for
    Geometry at Tel Aviv University. Kiril Solovey is also supported
    by the Clore Israel Foundation.}}

\author{Kiril Solovey \and Dan Halperin} 
\date{{\small Blavatnik School  of Computer Science, Tel Aviv University,
  Israel}}

\begin{document}

\maketitle
\thispagestyle{empty}
\pagestyle{empty}

\begin{abstract}
  We describe a general probabilistic framework to address a variety
  of \frechet-distance optimization problems.  Specifically, we are
  interested in finding minimal \emph{bottleneck}-paths in
  $d$-dimensional Euclidean space between given start and goal points,
  namely paths that minimize the maximal value over a continuous cost
  map. We present an efficient and simple sampling-based framework for
  this problem, which is inspired by, and draws ideas from, techniques
  for robot motion planning. We extend the framework to handle not
  only standard bottleneck pathfinding, but also the more demanding
  case, where the path needs to be monotone in all
  dimensions. Finally, we provide experimental results of the
  framework on several types of problems.
\end{abstract}


\section{Introduction}
This paper studies the problem of finding near-optimal paths in
$d$-dimensional Euclidean space. Specifically, we are interested in
\emph{bottleneck} paths which minimize the maximal value the path
obtains over a generally-defined continuous cost map.  As an example,
suppose that one wishes to plan a hiking route in a mountainous region
between two camping grounds, such that the highest altitude along the
path is minimized~\cite{BerKre97}. In this case, the map assigns to
each two-dimensional point its altitude. A similar setting, albeit
much more complex, requires to find a pathway of low energy for a
given protein molecule (see, e.g.,~\cite{RavETAL09}).

Our main motivation for studying bottleneck optimization over cost
maps is its tight relation to the \frechet distance (or matching), which is a
popular and widely studied similarity measure in computational
geometry. The problem has applications to various domains such as path
simplification~\cite{FilETAL14}, protein alignment~\cite{JiaXuZhu08},
handwritten-text search~\cite{SriKarBha07}, and signature
verification~\cite{ZheETAL08}. The \frechet distance, which was
initially defined for curves, is often considered to be a more
informative measure than the popular \emph{Hausdorff} distance as it
takes into consideration not only each curve as a whole but also the
location and the ordering of points along it. Usually one is
interested not only in the \frechet distance between two given curves,
but also in the parametrization which attains the optimal alignment.

Since its introduction by Alt and Godau~\cite{AltGod95} in 1995, a
vast number of works has been devoted to the subject, and many
algorithms have been developed to tackle various settings of the
problem. However, from a practical standpoint the problem is far from
being solved: for many natural extensions of the \frechet problem only
prohibitively-costly algorithms are known. Furthermore, in some cases
it was shown, via hardness proofs, that efforts for finding
polynomial-time algorithms are doomed to fail. For some variants of
the problem efficient algorithms are known to exist, however their
implementation requires complex geometric machinery that relies on
geometric kernels with infinite precision~\cite{MetMehPioSchYap08}.

\subparagraph*{Contribution.} We describe a generic, efficient and
simple algorithmic framework for solving pathfinding optimization
problems over cost maps. The framework is inspired by, and draws ideas
from, sampling-based methods for robot motion planning.  We provide
experimental results of the framework on various
scenarios. Furthermore, we theoretically analyze the framework and
show that the cost of the obtained solution converges to the optimum,
as the number of samples increases. We also consider the more
demanding case, where paths need to be monotone in all dimensions.

\subparagraph*{Organization.} In Section~\ref{sec:intro} we review
related work. In Section~\ref{sec:prel} we provide a formal definition
of the bottleneck pathfinding problem. In Section~\ref{sec:algorithms}
we describe an algorithmic framework for solving this problem. In
Section~\ref{sec:analysis} we provide an analysis of the
method. Finally, in Section~\ref{sec:experiments} we report on
experimental results.

\section{Related work}\label{sec:intro}
This section is devoted to related work on \frechet distance and robot
motion planning.

\subsection{\frechet distance} \label{sec:intro:frechet} The \emph{\frechet
  distance} between two curves is often described by an analogy to a
person walking her dog: each of the two creatures is required to walk
along a predefined path and the person wishes to know the length of
the shortest \emph{leash} which will make this walk possible. In many
cases one also likes to know how to advance along the path given the
short leash.

Formally, let $\sigma_1,\sigma_2:[0,1]\rightarrow \dR^d$ be two
continuous curves. We wish to find a traversal along the two curves
which minimizes the distance between the two traversal points. The
traversal is defined by two continuous parametrizations
$\alpha_1,\alpha_2:[0,1]\rightarrow [0,1]$ of $\sigma_1,\sigma_2$
respectively, where for a given point in time $\tau\in [0,1]$, the
positions of the person and her dog are specified by
$\sigma_1(\alpha_1(\tau))$ and $\sigma_2(\alpha_2(\tau))$,
respectively. The \emph{\frechet distance} between $\sigma_1,\sigma_2$
is defined by the expression
$$\min_{\alpha_1,\alpha_2:[0,1]\rightarrow [0,1]}\max_{\tau\in [0,1]}
\|\sigma_1(\alpha_1(\tau))-\sigma_2(\alpha_2(\tau))\|_2.$$

Alt and Godau~\cite{AltGod95} described an $O(n^2\log n)$-time
algorithm for the setting of two polygonal curves, where $n$ is the
number of vertices in each of the two curve. Buchin et
al.~\cite{BucBucETAL13} described a different method for solving this
problem for the same running time.  Recently, Buchin et
al.~\cite{BucBucMeuMul14} developed an algorithm with a slightly
improved running time $O(n^2\log^2\log n)$. Har-Peled and
Raichel~\cite{HarRai14} introduced a simpler randomized algorithm with
running time of $O(n^2\log n)$.  Bringmann~\cite{Bri14} showed that an
algorithm with running time of $O(n^{2-\delta})$, for some constant
$\delta >0$, does not exist, unless a widely accepted conjecture,
termed SETH~\cite{ImpPatZan01}, is wrong. In a following
work~\cite{BriMul15} this conditional lower bound was extended to
$(1+\varepsilon)$-approximation algorithms of the \frechet problem,
where $\varepsilon\leq 0.399$.

The notion of \frechet distance can be extended to $k$ curves in
various ways. One natural extension can be described figuratively as
having a pack of $k$ dogs, where each of the dogs has to walk along a
predefined path, and every pair of dogs is connected with a leash.
The goal now is to find a parametrization which minimizes the length
of the longest leash. Dumitrescu and Rote~\cite{DumRot04} introduced a
generalization of the Alt-Godau algorithm to this case, which runs in
$O(kn^k\log n)$ time, i.e., exponential in the number of input
curves. They also describe a $2$-approximation algorithm with a much
lower running time of $O(k^2n^2\log n)$. In the work of Har-Peled and
Raichel~\cite{HarRai14} mentioned above they also consider the case of
$k$ input curves and devise an $O(n^k)$ algorithm. Notably, their
technique is flexible enough to cope with different \frechet-type goal
functions over the $k$ curves. Furthermore, their algorithm is also
applicable when the $k$ curves are replaced with $k$ \emph{simplicial
  complexes}, and the problem is to find $k$ curves---one in each
complex---which minimize the given goal function. A recent
work~\cite{BucETAL16}, which extends the conditional lower bound
mentioned earlier for the setting of multiple curves, suggests that a
running time that is exponential in the number of curves is
unavoidable.

The notion of \frechet distance can be generalized to more complex
objects. Buchin et al.~\cite{BucBucWen08} considered the problem of
finding a mapping between two simple polygons, which minimizes the
maximal distance between a point and its image in the other
polygon. More formally, given two simple polygons $P,Q\subset \dR^2$
the problem consists of finding a mapping $\delta:P\rightarrow Q$
which minimizes the expression $\max_{p\in P}\|p-\delta(p)\|_2$, subject
to various constraints on $\delta$.  They introduced a polynomial-time
algorithm for this case. In a different paper, Buchin et
al.~\cite{BucBucSch08} showed that the decision problem is \np-hard
for more complex geometric objects, e.g., pairs of polygons with holes
in the plane or pairs of two-dimensional terrains. Another interesting
\np-hard problem that was studied by Sherette and Wenk~\cite{SheWen13}
is \emph{curve embedding} in which one wishes to find an embedding of
a curve in $\dR^3$ to a given plane, which minimizes the \frechet
distance with the curve. In a similar setting Meulemans~\cite{Meu13}
showed that it is \np-hard to decide whether there exists a simple
cycle in a plane-embedded graph that has at most a given \frechet
distance to a simple closed curve.

The \frechet distance between curves in the presence of obstacles have
earned some attention. Cook and Wenk~\cite{CooWen10} studied the
\emph{geodesic} variant, which consists of a simple polygon and two
polygonal curves inside it. As in the standard formulation, the main
goal is to minimize the length of the leash, but now the leash may
wrap or bend around obstacles. Their algorithm has running time of
$O(m+n^2\log mn\log n)$, where $m$ is the complexity of the polygon
and $n$ is defined as the total complexity of the two curves, as
before.  The more complex \emph{homotopic} setting is a special case
of the aforementioned geodesic setting, with the additional constraint
that the leash must continuously deform.  Chambers et
al.~\cite{ChaETAL10} considered this problem for the specific setting
of two curves in planar environment with polygonal obstacles. They
developed an algorithm whose running time is $O(N^9\log N)$, where
$N=n+m$ for $n$ and $m$ as defined above.

\subsection{Motion planning}
Motion planning is a fundamental problem in robotics. In its most
basic form, the problem consists of finding a collision-free path for
a robot $\R$ in a workspace environment $\W$ cluttered with
obstacles. Typically, the problem is approached from the configuration
space $\C$---the set of all robot configurations. The problem can be
reformulated as finding a continuous curve in $\C$, which entirely
consists of collision-free configurations and represents a path for
the robot from a given start configuration to another, target,
configuration. An important attribute of the problem is the number of
\emph{degrees of freedom} of $\R$, using which one can specify every
configuration in $\C$. Typically the dimension of $\C$ equals the
number of degrees of freedom.

For some cases of the problem, which involve a small number of degrees
of freedom, efficient and exact analytical techniques exist (see,
e.g., ~\cite{AvnBoiFav88, HalSha96, Sharir04}), which are guaranteed
to find a solution if one exists, or report that none exists
otherwise. Recently, it was shown~\cite{SolYuZamHal15,
  abhs-unlabeled14, tmk-cap13} that efficient and complete techniques
can be developed for the \emph{multi-robot} motion-planning problem,
which entails many degrees of freedom, by making several simplifying
assumptions on the separation of the start and target
positions. However, it is known that the general setting of the
motion-planning problem is computationally intractable (see,
e.g.,~\cite{Rei79,hss-cmpmio,sy-snp84,SolHal15}) with respect to the
number of degrees of freedom.

Sampling-based algorithms for motion planning, which were first
described about two decades ago, have revolutionized the field of
robotics by providing simple yet effective tools to cope with
challenging problems involving many degrees of freedom. Such
algorithms (see, e.g., PRM by Kavraki et al.~\cite{kslo-prm}, RRT by
Kuffner and LaValle~\cite{l-rert}, and EST by Hsu et
al.~\cite{HsuLatMot99}) explore the high-dimensional configuration
space by random sampling and connecting nearby samples, which result
in a graph data structure that can be viewed as an approximation of
the \emph{free space}---a subspace of $\C$, which consists entirely of
collision-free configurations.  While such techniques have weaker
theoretical guarantees than analytical methods, many of them are
\emph{probabilistically complete}, i.e., guaranteed to find a solution
if one exists, given sufficient processing time. More recently,
\emph{asymptotically optimal} sampling-based algorithms, whose
solution converges to the optimum, for various criteria, have started
to emerge: Karaman and Frazzoli introduced the RRT* and
PRM*~\cite{kf-sbaomp11} algorithms, which are asymptotically optimal
variants of RRT and PRM. Following their footsteps Arslan and Tsiotras
introduced RRT\#~\cite{ArsTsi13}. A different approach was taken by
Janson and Pavone who introduced the FMT* algorithm~\cite{JanPav13},
which was later refined by Salzman and Halperin~\cite{SalHal14a}.

\section{Problem statement}\label{sec:prel}
In this section we describe the general problem of bottleneck
pathfinding over a given cost map, to which we describe an algorithmic
framework in Section~\ref{sec:algorithms}. We conclude this section we
several concrete examples of the problems that will be used for
experiments in Section~\ref{sec:experiments}. 

We start with several basic definitions. Given $x,y\in\dR^d$, for some
fixed dimension $d\geq 2$, let $\|x-y\|_2$ denote the Euclidean
distance between two points.  Denote by $\B_{r}(x)$ the
$d$-dimensional Euclidean ball of radius $r>0$ centered at
$x\in \dR^d$ and $\B_{r}(\Gamma) = \bigcup_{x \in \Gamma}\B_{r}(x)$
for any $\Gamma \subseteq \dR^d$. We will use the terms ``path'' and
``curve'' interchangeably, to refer to a continuous curve in $\dR^d$
parametrized over~$[0,1]$. Given a curve
$\sigma:[0,1]\rightarrow \dR^d$ define
$\B_r(\sigma)=\bigcup_{\tau\in[0,1]}\B_r(\sigma(\tau))$. Additionally,
denote the image of a curve $\sigma$ by
$\Im(\sigma)=\bigcup_{\tau\in[0,1]}\{\sigma(\tau)\}$.  Let
$A_1,A_2,\ldots$ be random variables in some probability space and let
$B$ be an event depending on $A_n$. We say that $B$ occurs
\emph{almost surely} (\as, in short) if \mbox{$\limn\Pr[B(A_n)]=1$}.

Let $\M:[0,1]^d\rightarrow \dR$ be a \emph{cost map} that assigns to
each point in $[0,1]^d$ a real value. For simplicity, we assume that
the domain of $\M$ is a $d$-dimensional unit hypercube. Let
$S,T\in [0,1]^d$ denote the start and target points. Denote by
$\Sigma(S,T)$ the collection of paths that start in $S$ and end in
$T$. Formally, every $\sigma\in \Sigma(S,T)$ is a continuous path
$\sigma:[0,1]\rightarrow[0,1]^d$, where
\mbox{$\sigma(0)=S,\sigma(1)=T$}. Given a path $\sigma$ we use the
notation $\M(\sigma)=\max_{\tau\in [0,1]}\M(\sigma(\tau))$ to
represent its bottleneck cost.

In some applications, monotone paths are desired. For instance, in the
classical problem of \frechet matching between two curves it is often
the case that backward motion along the curves is forbidden. Here we
consider monotonicity in all $d$ coordinates of points along the
path. Formally, given two points $p,p'\in \dR^d$, where
$p=(p_1,\ldots,p_d),p'=(p'_1,\ldots,p'_d)$, we use the notation
$p\preceq p'$ to indicate that $p_i\leq p'_i$, for every
$1\leq i\leq d$. A path $\sigma\in \Sigma(S,T)$ is said to be
\emph{monotone} if for every \mbox{$0\leq \tau\leq \tau'\leq 1$} it
holds that $\sigma(\tau) \preceq \sigma(\tau')$.

\begin{definition}\label{def:bpp}
  Given the triplet $\langle \M, S,T\rangle$, the
  \emph{bottleneck-pathfinding problem} (BPP, for short) consists of
  finding a path $\sigma\in \Sigma(S,T)$ which minimizes the
  expression $\max_{\tau\in [0,1]}\M(\sigma(\tau))$. A special case of
  the bottleneck pathfinding problem, termed \emph{strong}-BPP,
  requires that the path will be \emph{monotone}.
\end{definition}

\subsection{Examples}\label{sec:prel:examples}
We provide three examples of BPPs, which will be used for experiments
in Section~\ref{sec:experiments}. Each example is paired with the
$d$-dimensional configuration space $\C:=[0,1]^d$, start and target
points $S,T\in \C$, and a cost map $\M:[0,1]^d\rightarrow \dR$. The
examples below are defined for two-dimensional input objects, but can
generalized to higher dimensions. 

\subparagraph*{Problem 1:} We start with the classical
\emph{\frechet distance among $k$~curves} (see,
e.g.,~\cite{HarRai14}). Let
$\sigma_1,\ldots,\sigma_k :[0,1]\rightarrow [0,1]^2$ be $k$ continuous
curves embedded in Euclidean plane. Here $\C=[0,1]^k$ is defined as
the Cartesian product of the various positions along the $k$
curves. Namely, a point $P=(p_1,\ldots,p_k)\in \C$ describes the
location $\sigma_i(p_i)$ along $\sigma_i$, for each $1\leq i \leq k$.
To every such $P$ we assign the cost
$\M(P)=\max_{1\leq i<j\leq k}\|\sigma_i(p_i)-\sigma_j(p_j)\|_2$. We
note that more complex formulations of $\M$ can be used, depending on
the exact application. The start and target positions are defined to
be
$S=(\sigma_1(0),\ldots, \sigma_k(0)), T=(\sigma_1(1),\ldots,
\sigma_k(2))$.

\subparagraph*{Problem 2:} We introduce the problem of
\emph{\frechet distance with visibility}, whose basis is similar
to~\textbf{P1} with $k=3$. In addition to the curves, we are given a
subspace $\F\subseteq [0,1]^2$. The goal is to find a traversal of the
curves which minimizes $\M$ as defined in \textbf{P1}, with the
additional constraint that the traversal point along $\sigma_1$ must
be ``seen'' by one of the traversal points of
$\sigma_2,\sigma_3$. Formally, for every $P=(p_1,p_2,p_3)\in \C$ it
must hold that $p_1p_2\subset \F$ \emph{or} $p_1p_3\subset\F$ (but not
necessarily both), where $p_ip_j$ is the straight-line path from $p_i$
to $p_j$.

\subparagraph*{Problem 3:} In \emph{curve embedding}
(see,~\cite{SheWen13, Meu13}), the input consists of a curve
$\sigma:[0,1]\rightarrow [0,1]^2$, a subspace $\F\subseteq [0,1]^2$
and a pair of two-dimensional points $s,t\in \F$. A point
$P=(p_1,p_2,p_3)\in\C=[0,1]^3$ describes the location $\sigma(p_1)$
along $\sigma$ and the point $(p_2,p_3)\in \F$.  The BPP is defined
for the start and target points $S=(0,s), T=(1,t)\in \C$ and the cost
map $\M(P)=\|\sigma(p_1)-(p_2,p_3)\|_2$.

\section{Algorithmic framework}\label{sec:algorithms}
In this section we describe an algorithmic framework that will be used
for solving standard and strong regimes of BPP
(Definition~\ref{def:bpp}). The framework can be viewed as a variant
of the PRM algorithm~\cite{kf-sbaomp11}, and we chose to describe it
here in full detail for completeness. However, the analysis provided
in Section~\ref{sec:analysis} is brand new.

The framework consists of three conceptually simple steps: In the
first step, we construct a random graph embedded in $[0,1]^d$, whose
vertices consist of the start and target points $S,T$, and of a
collection of randomly sampled points; the edges connect points that
are separated by a distance of at most a given connection threshold
$r_n$.  In the second step the edges of the graph are assigned with
weights corresponding to their bottleneck cost over $\M$. In the third
and final step, the discrete graph is searched for a path connecting
$S$ to $T$ which minimizes the bottleneck cost.

Before proceeding to a more elaborate description of the framework we
provide a formal definition of the random graphs that are at the heart
of the technique. Let $\X_n=\{X_1,\ldots,X_n\}$ be $n$ points chosen
independently and uniformly at random from the Euclidean
$d$-dimensional cube~$[0,1]^d$. The following definition corresponds
to the standard and well-studied model of random geometric
graphs~(see, e.g.,~\cite{Pen03, Wal11, BalETAL08} and the literature
review in~\cite{SolETAL16}).

\begin{definition}\label{def:rgg}
  The \emph{random geometric graph} (RGG) $\G_n = \G(\X_n;r_n)$ is a
  \emph{directed} graph with vertex set $ \X_n$ and edge set
  $\{(x,y): x \neq y, \ x,y \in \X_n, \ \|x-y\|_2\leq r_n\}$.
\end{definition}

We are ready to describe the framework, which has two parameters: $n$
represents the number of samples generated and $r_n$ defines the
Euclidean connection radius used in the construction of the graphs. In
the next section we show that for a range of values of $r_n$, which is
a function of the number of samples $n$, the cost of the returned
solution converges to the optimum, as $n$ tends to infinity. The
framework consists of the following steps:

\subparagraph*{Step I:} We construct the RGG
\mbox{$\G_n=(\X_n\cup \{S,T\};r_n)$}. For the purpose of generating
$\G_n$ a collection of $n$ samples $\X_n$ is generated and a
nearest-neighbor structure is employed to find for every
$x\in \X_n\cup \{S,T\}$ the set of samples that located within a
Euclidean distance of $r_n$ from it.

\subparagraph*{Step II:} We assign to each edge of the graph the
bottleneck cost of the straight-line path connecting its endpoints
under $\M$. In particular, for the standard BPP, for every edge
$(x,y)$ the cost $\max_{\tau\in[0,1]}\M(x+\tau(y-x))$ is assigned. The
same applies for strong-BPP, unless $x\not\preceq y$, in which case
the value $+\infty$ is assigned. 

\subparagraph*{Step III:} For the final step we find a path over
$\G_n$ from $S$ to $T$ which minimizes the bottleneck cost. Several
efficient algorithms solving this problem exist (see,
e.g.,~\cite{Wil08,CheETAL16}).

\section{Theoretical foundations}\label{sec:analysis}
We study the behavior of the framework for the standard and the strong
case of BPP (Definition~\ref{def:bpp}). Recall the framework uses the
two parameters $n$ and $r_n$, which specify the number of samples and
the connection radius.We establish a range of connection radii, $r_n$,
for which the cost of the returned solution is guaranteed to converge
to a relaxed notion of the optimum.

The analysis below does not restrict itself to a specific type of cost
maps $\M$, e.g., continuous or smooth. Thus, due to the stochastic
nature of the framework, and the general definition of $\M$, we cannot
guarantee that the returned solution will tend to the absolute
optimum.  As an example consider the cost map $\M$ such that for a
given $x=(x_1,x_2)$, $\M(x)=0$ if $x_1=x_2$, and $\M(x)=1$
otherwise. For the start and target points $S=(0.1,0.1),T=(0.9,0.9)$
the optimal solution is a subset of the diagonal. Obviously, the
probability of having a single point of $\X_n$, let alone a whole path
in $\G_n$, that lie on the diagonal is equal to $0$.

We can however guarantee convergence to a \emph{robustly-optimal}
path, which is defined below.  Informally, such paths have
``well-behaved'' neighborhoods, in terms of the value of $\M$. We
provide below a formal definition of this notion for the bottleneck
cost function. Recall that given a path $\sigma$ the notation
$\M(\sigma)$ represents its bottleneck cost.

\begin{definition}
  Given the triplet $\langle \M, S,T\rangle$, a path
  $\sigma\in \Sigma(S,T)$ is called \emph{robust} if for every
  \mbox{$\eps>0$} there exists \mbox{$\delta>0$} such that
  $\M(\sigma')\leq (1+\eps)\M(\sigma)$, for any
  $\sigma'\in \Sigma(S,T)$ such that
  \mbox{$\Im(\sigma')\subset \B_\delta(\sigma)$}. A path that attains
  the infimum cost, over all robust paths, is termed \emph{robustly
    optimal}.
\end{definition}

\subsection{(Standard) Bottleneck cost}\label{sec:analysis_weak}
For a given triplet $\langle \M, S,T\rangle$ representing an instance
of BPP, denote by $\sigma^*$ a robustly-optimal solution. Note that
we do not require here that $\sigma^*$ or the returned solution will
be monotone. We obtain the following result. All logarithms stated
henceforth are to base~$e$.

\begin{theorem}\label{thm:main_weak}
  Let $\G_n=\G(\X_n\cup\{S,T\};r_n)$ be an RGG with
  $$r_n=\gamma\left(\frac{\log n}{n}\right)^{1/d},\quad
  \gamma>2(2d\theta_d)^{-1/d},$$
  where $\theta_d$ denotes the Lebesgue measure of a unit ball in
  $\dR^d$.  Then $\G_n$ contains a path $\sigma_n\in \Sigma(S,T)$ such
  that $\M(\sigma_n)=(1+o(1))\M(\sigma^*)$, \as
\end{theorem} 

We mention that this connection radius is also essential for
connectivity of RGGs, i.e., a smaller radius results in a graph that
is disconnected with high probability (see,
e.g.,\cite{BroETAL14}). This fact is instrumental to our proof.  We
also mention that a result similar to Theorem~\ref{thm:main_weak} can
be obtained through a different proof technique~\cite{kf-sbaomp11},
albeit with a larger value of the constant $\gamma$.

For simplicity, we assume for the purpose of the proof that exists a
finite constant $\delta'>0$ such that
$\B_{\delta'}(\sigma^*)\subset [0,1]^d$, namely the robustly-optimal
solution is at least $\delta'$ away from the boundary of the domain
$[0,1]^d$.  This constraint can be easily relaxed by transforming
$\langle \M, S,T\rangle$ into an equivalent instance
$\langle \M', S',T'\rangle$ where this condition is met. In particular
the original input can be embedded to a cube of side length $1-\eps$
for some constant $\eps>0$, which is centered in the middle of
$[0,1]^d$. The cost along the boundaries of the smaller cube should be
extended to the remaining parts of the $[0,1]^d$ cube.

Given an RGG $\G_n$ and a subset $\Gamma\subset [0,1]^d$ denote by
$\G_n(\Gamma)$ the graph obtained from the intersection of $\G_n$ and
$\Gamma$: it consists of the vertices of $\G_n$ that are contained in
$\Gamma$ and the subset of edges of $\G_n$ that are fully contained in
$\Gamma$. Each edge is considered as a straight-line segment
connecting its two end points.

A main ingredient in the proof of Theorem~\ref{thm:main_weak} is the
following Lemma. We employ the \emph{localization-tessellation}
framework~\cite{SolETAL16}, which was developed by the authors. The
framework allows to extend properties of RGGs to domains with complex
geometry and topology. 

\begin{lemma}\label{lem:rgg_paper}
  Let $\G_n$ be the RGG defined in
  Theorem~\ref{thm:main_weak}. Additionally let
  $\Gamma\subset [0,1]^d$ be a fixed subset, where $S,T\in \Gamma$,
  and let $\rho>0$ be some fixed constant, such that
  $\B_\rho(\Gamma)\subset [0,1]^d$. Then $S,T$ are connected in
  $\G_n(\B_\rho(\Gamma))$ \as
\end{lemma}
\begin{proof}
  We rely on the well-known result that $\G_n$ is connected \as in the
  domain $[0,1]^d$ for the given connection radius $r_n$ (see,
  e.g.,~\cite[Theorem~1]{SolETAL16}). We then use Lemma~1 and Theorem~6
  in~\cite{SolETAL16} which state that if $\G_n$ is connected \as, and
  is \emph{localizable} (see, Definition~6 therein), then $S,T$ are 
  connected \as over $\G_n(\B_\rho(\Gamma))$.
\end{proof}

\begin{proof}[Proof of Theorem~\ref{thm:main_weak}]
  We first show that for any $\eps>0$ it follows that
  $\M(\sigma_n)\leq (1+\eps)\M(\sigma^*)$ \as Fix some $\eps>0$. Due
  to the fact that $\sigma^*$ is robustly optimal, there exists
  $\delta_\eps>0$ independent of $n$ such that for every
  $\sigma\in\Sigma(S,T)$ such that
  $\Im(\sigma)\subset \B_{\delta_\eps}(\sigma^*)$ we have that
  $\M\leq (1+\eps)\M(\sigma^*)$ \as Additionally, recall that there
  exists some $\delta'>0$ such
  $\B_{\delta'}(\sigma^*)\subset [0,1]^d$.

  Set $\delta=\min\{\delta_\eps,\delta'\}$ and define the sets
  $\Gamma_{\delta/2}=\B_{\delta/2}(\sigma^*),
  \Gamma_{\delta}=\B_{\delta}(\sigma^*)$
  and notice that $S,T\in \Gamma_{\delta/2}$. By
  Lemma~\ref{lem:rgg_paper} we have that $S,T$ are connected in
  $\G_n(\Gamma_\delta)$.  Moreover, a path connecting $S,T$ in
  $\G_n(\Gamma_{\delta})$ must a have a bottleneck cost of at most
  $(1+\eps)\M(\sigma^*)$.

  We have shown that for any fixed $\eps>0$,
  $\M(\sigma_n)\leq (1+\eps)\M(\sigma^*)$ \as By defining the sequence
  $\eps_i=1/i$ one can extend the previous result and show that
  $\M(\sigma_n)\leq (1+o(1))\M(\sigma^*)$.  This part is technical and
  its details are omitted (see a similar proof
  in~\cite[Theorem~6]{ssh-fne13}). This concludes the proof.
\end{proof}

\subsection{Strong bottleneck cost}\label{sec:analysis_strong}
We now focus on the strong case of the problem, where the solution is
restricted to paths that are monotone in each of the $d$
coordinates. Denote by $\vec\sigma^*$ the robustly-optimal monotone
solution for a given instance $\langle \M, S,T\rangle$. 

\begin{theorem}\label{thm:main_strong}
  Let $\G_n=\G(\X_n\cup\{S,T\};r_n)$ be an RGG with
  $r_n=\omega(1)\left(\frac{\log n}{n}\right)^{1/d}.$ Then $\G_n$
  contains a monotone path $\vec\sigma_n\in \Sigma(S,T)$ such that
  $\M(\vec\sigma_n)=(1+o(1))\M(\vec\sigma^*)$, \as
\end{theorem}

Let $x,x'\in [0,1]^d$ be two points such that $x \preceq x'$. For a
given $\delta>0$ the notation $x\preceq_\delta x'$ indicates that
$\delta=\min\{x'_i-x_i\}_{i=1}^d$, where
$x=(x_1,\ldots,x_d),x'=(x'_1,\ldots,x'_d)$.  Given two points
$x,x'\in [0,1]^d$, such that $x\preceq x'$, denote by $\H(x,x')$ the
$d$-dimensional box
$\left[x_1,x'_1\right]\times\ldots \times\left[x_d,x'_d\right]$. In
addition to the assumption that the robustly-optimal solution
$\vec\sigma^*$ is separated from the boundary of $[0,1]^d$ that we
have taken in the previous analysis, we also assume that there exists
a constant $0<\delta'' \leq 1$ such that $S\preceq_{\delta''}T$.

In preparation for the main proof we prove the following lemma.

\begin{lemma}\label{lem:directed_localization}
  Choose any\footnote{For instance, $f_n$ can be either one of the
    following functions: $\log n,\log^*n$, or the inverse Ackerman
    function $\alpha(n)$.}  $f_n\in \omega(1)$ and set
  $r_n=\omega(1)\left(\frac{\log n}{n}\right)^{1/d}$.  Let
  $q,q'\in [0,1]^d$ be two points such that $q\preceq_\delta q'$,
  where $\delta$ is independent of $n$. Then \as~there exist
  $X,X'\in \X_n$ with the following properties: (i)
  $\|X-q\|_2\leq r_n/2,\|X'-q'\|\leq r_n/2$; (ii)
  $q\preceq X,X'\preceq q'$; (iii) $X,X'$ are connected in $\G_n$ with
  a monotone path.
\end{lemma}

\begin{proof}
  We apply a tessellation argument similar to the one used to show
  that the standard (and undirected) RGG is connected (see,
  e.g.,~\cite[Section~2.4]{Wal11}).  Set
  $\ell = \left\lceil\frac{2\|q'-q\|_2}{r_n}\right\rceil$ and observe
  that $\ell\leq 2\sqrt{d}/r_n$. Define the normalized vector
  $\vec v=\frac{q'-q}{\|q'-q\|_2}$ and let $H_1,\ldots,H_\ell$ be a
  sequence of $\ell$ hyperboxes, where
  \[H_j=\H\left(q+(j-1)\cdot \frac{r_n}{2}\cdot \vec v, q+j\cdot
    \frac{r_n}{2}\cdot \vec v\right),\]
  for every $1\leq j\leq \ell$ (see Figure~\ref{fig:localization}).
  Observe that for every $1\leq j<\ell$ and every
  $X_j\in H_j,X_{j+1}\in H_{j+1}$, we have
  \begin{equation}
    X_j\preceq X_{j+1}, \|X_{j+1}-X_j\|_2\leq r_n.\label{eq:1}
  \end{equation}

  We show that for every $1\leq j\leq \ell$ it follows that
  $\X_n\cap H_j\neq \emptyset$, \as We start by bounding the volume of
  $H_j$. Denote by $c_1,\ldots,c_d$ the side lengths of $H_j$, and
  denote by $\delta_1,\ldots,\delta_d$ the side lengths of
  $\H(q,q')$. Note that $\delta_i$ is independent of $n$ and
  $c_i=\delta_i/\ell$. Consequently, we can represent
  $c_i = \alpha_i r_n$, where $\alpha_i>0$ is constant, for every
  $1\leq i\leq d$. Thus, $|H_j|=c r_n^d$ for some constant $c>0$.
  Now,
  \begin{align*}
    \Pr\left[ \X_n\cap H_j=\emptyset \right] = (1-|H_j|)^n\leq
    \exp\left\{-n|H_j|\right\} =
    \exp\left\{-\omega(1)\cdot c
    \log
    n\right\}
    \leq n^{-1}.
  \end{align*}
  In the last transition we used the fact that the function
  $f_n\in \omega(1)$ can ``absorb'' any constant $c$. We are ready to
  show that every $H_i$ contains a point from $\X_n$ \as:
  \begin{align*}
    \Pr\left[\exists H_j: \X_n\cap H_j=\emptyset \right] & \leq
    \sum_{j=1}^\ell\Pr\left[ \X_n\cap H_i=\emptyset \right] \\ & \leq
    \ell \cdot n^{-1} \leq \frac{2\sqrt{d}}{r_n}\cdot  n^{-1}  =
    \frac{2\sqrt d }{\omega(1)\cdot n^{1-1/d}\log^{1/d} n}.
  \end{align*}
  Thus, \as~there exists for every $1\leq j\leq \ell$ a point
  $x_j\in H_j$. Observe that $X:=X_1,X':=X_\ell$ satisfy
  (i),(ii). Condition (iii) follows from Equation~\ref{eq:1}.
\end{proof}

\begin{figure}
  \centering
  \includegraphics[width=0.5\textwidth]{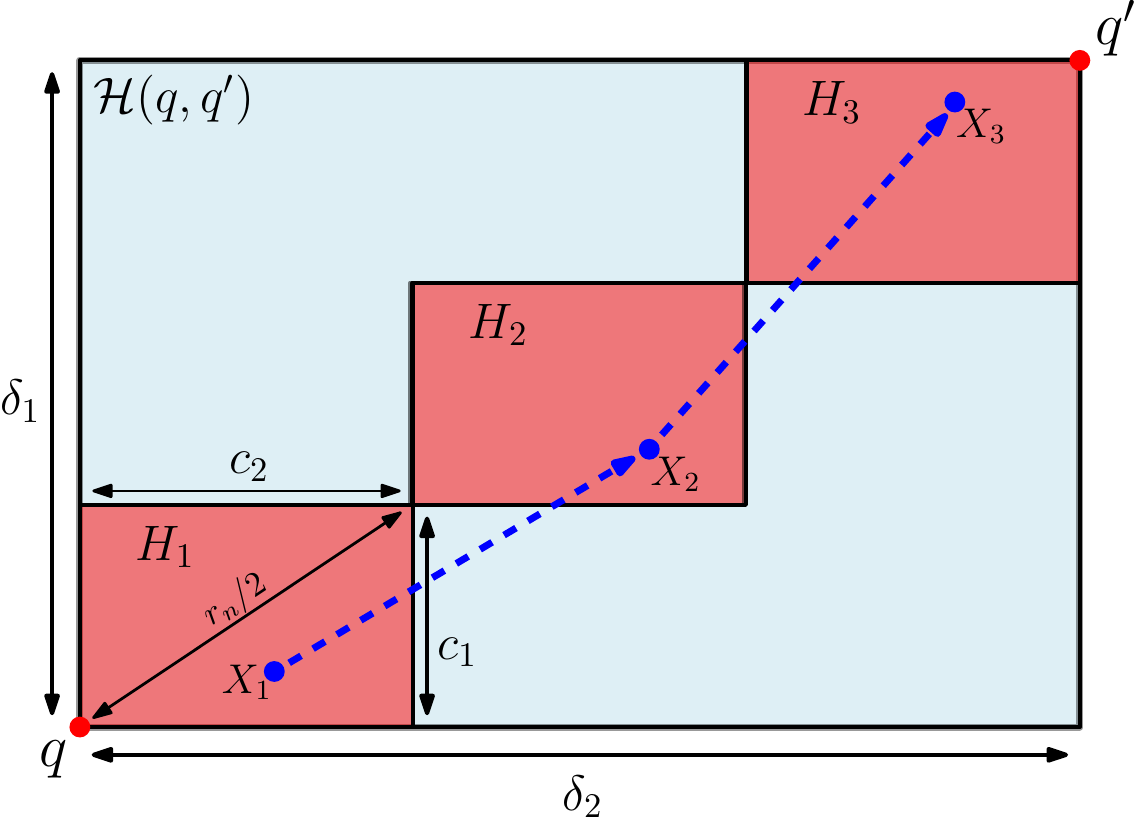}
  \caption{Visualization of the proof of
    Lemma~\ref{lem:directed_localization} for $d=2$. The blue
    rectangle represents $\H(q,q')$ and the three red rectangles
    represent $H_1,\ldots,H_\ell$ for $\ell=3$ (the small value of
    $\ell$ was selected for the clarity of visualization and in
    reality $r_n\ll\delta_1$).  The length of the largest diagonal in
    each of the small rectangles is $r_n/2$, which implies that a
    distance between $X_j\in H_j,X_{j+1}\in H_{j+1}$ is at most
    $r_n$.  The blue dashed arrows represent the directed graph
    edges $(X_1,X_2),(X_2,X_3)$ which correspond to a monotone path
    connecting $X_1$ to $X_3$. }
  \label{fig:localization}
\end{figure}

\begin{proof}[Proof of Theorem~\ref{thm:main_strong}]
  Similarly to the proof of Theorem~\ref{thm:main_weak}, we fix
  $\eps>0$ and select $\delta\leq \min\{\delta',\delta''\}$ such that
  $\M(\vec\sigma)\leq (1+\eps)\M(\vec\sigma^*)$ for every
  $\vec\sigma\in \vec\Sigma(S,T)$ with
  $\Im(\vec\sigma) \subset \B_\delta(\vec\sigma^*)\subset [0,1]^d$.

  The crux of this proof is that there exists a sequence of $k$ points
  $q_1,\ldots,q_k\in \Im(\vec\sigma^*)$, where $S=q_1,T=q_k$, such
  that $q_j\prec_{\delta/2}q_{j+1}$ for every $1\leq j<k$ (see
  Figure~\ref{fig:thm_directed}).  Moreover, due to fact that
  $\vec\sigma^*$ is monotone we can determine that such $k$ is finite
  and independent of $n$. Thus, by
  Lemma~\ref{lem:directed_localization}, for every $1\leq j<k$ there
  exist $X_j,X'_{j}\in \X_n$ which satisfy the following conditions
  \as: (i) $\|X_j-q_j\|_2\leq r_n/2,\|X'_j-q_{j+1}\|_2\leq r_n/2$;
  (ii) $q_j\preceq X_j,X'_j\preceq q_{j+1}$; (iii) $X_j,X'_j$ are
  connected in $\G_n$. By conditions (i),(ii), for every $1\leq j<k$
  the graph $\G_n$ contains the edge $(X'_j,X_{j+1})$. Combined with
  condition (iii) this implies that $S$ is connected to $T$ in $\G_n$
  \as

  It remains to show that the path constructed above has a cost of at
  most $(1+\eps)\M(\vec\sigma^*)$. For every $1\leq j\leq k$ denote by
  $\vec\sigma_j$ the path induced by
  Lemma~\ref{lem:directed_localization} from $X_j$ to $X'_j$, i.e.,
  $\vec\sigma_j(0)=X_j,\vec\sigma_j(1)=X'_j$ and
  $\Im(\vec\sigma_j)\subset H_i$. Additionally, for every $1\leq j<k$
  denote by $\vec\sigma'_j$ the straight-line segment (sub-path) from
  $X'_j$ to $X_{j+1}$. Now, define $\vec\sigma$ to be a concatenation
  of
  $\vec\sigma_1,\vec\sigma'_1,\ldots,
  \vec\sigma_{k-1},\vec\sigma'_{k-1},\vec\sigma_k$.
  We showed in the previous paragraph that such a path exists in
  $\G_n$ \as Observe that for every $1\leq j\leq k$ it holds that
  $\vec\sigma_j\subset \H(q_j,q_{j+1})$, where
  $\H(q_j,q_{j+1})\subset \B_{\delta/2}(\vec\sigma^*)$. This implies
  that \mbox{$\M(\vec\sigma_i)\leq (1+\eps)\M(\vec\sigma^*)$}.
  Additionally, recall that for every $1\leq j<k$ it holds that
  $\|X'_j-q_{j+1}\|_2\leq r_n/2, \|X_{j+1}-q_{j+1}\|_2\leq r_n/2$,
  which implies that
  $\Im(\vec\sigma'_j)\subset \B_{r_n}(q_{j+1})\subset
  \B_{\delta}(\vec\sigma^*)$,
  and consequently
  \mbox{$\M(\sigma'_j)\leq (1+\eps)\M(\vec\sigma^*)$}.  Finally,
  $\M(\vec\sigma_n)\leq \M(\vec\sigma)\leq
  (1+\eps)\M(\vec\sigma^*)$. This concludes the proof.
\end{proof}

\begin{figure}
  \centering
  \includegraphics[width=0.5\textwidth, trim=100pt 50pt 100pt 140pt , clip=true]{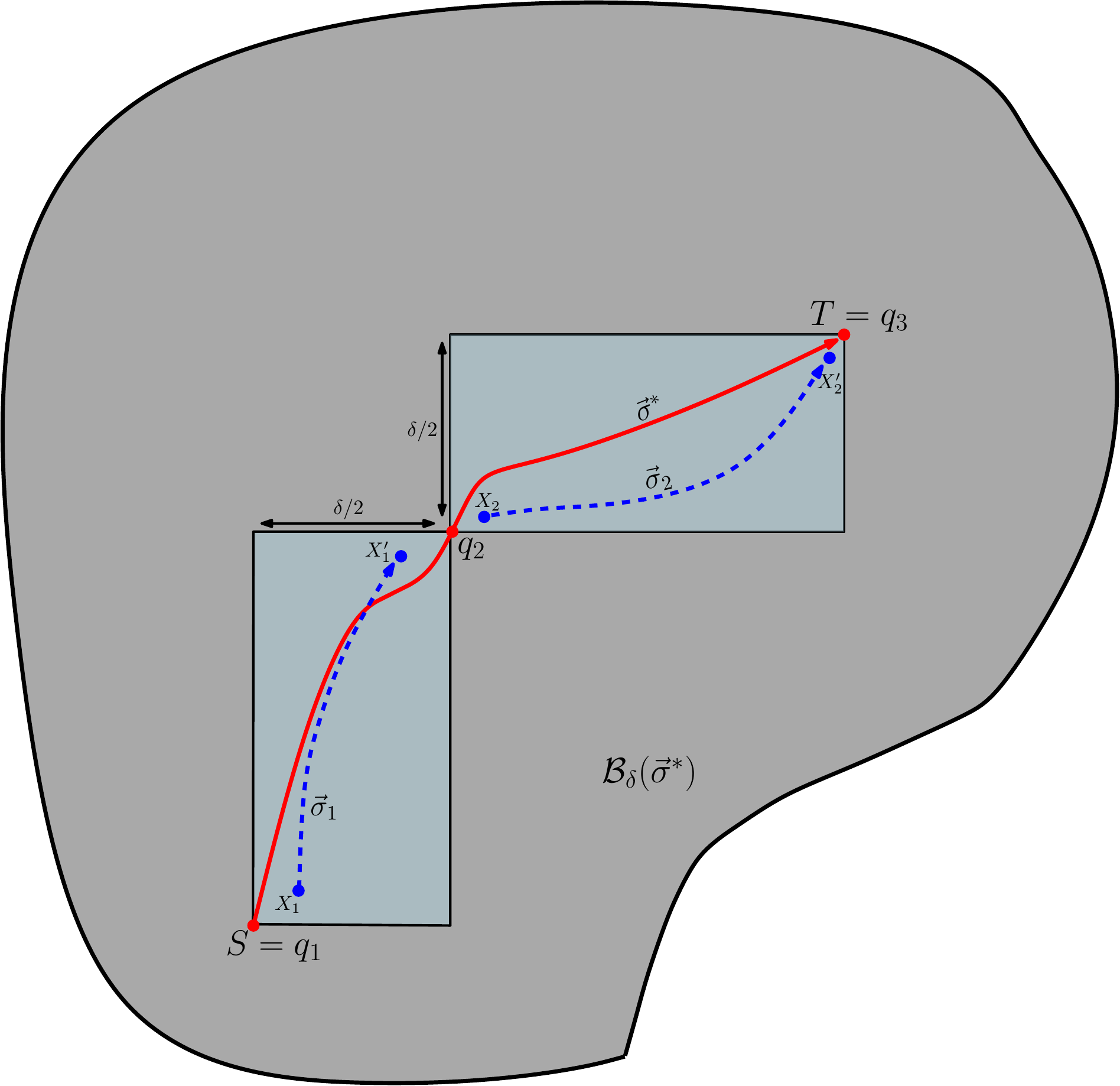}
  \caption{Visualization of the proof of Theorem~\ref{thm:main_strong}
    for $d=2$ and $k=3$. The red curve represents $\vec\sigma^*$, on
    which lie the points $q_1,q_2,q_3$ such that
    $q_1\prec_{\delta/2} q_2 \prec_{\delta/2} q_3$. The dashed blue
    curves represent $\vec\sigma_1,\vec\sigma_2$. The gray area
    represents $\B_\delta(\vec\sigma^*)$. }
  \label{fig:thm_directed}
\end{figure}

\section{Experimental results}\label{sec:experiments}
In this section we validate the theoretical results that were
described in the previous section. We observe that the framework can
cope with complex scenarios involving two or three degrees of freedom
($d\in \{2,3\}$), and converges quickly to the optimum.

Before proceeding to the results we provide details regrading the
implementation. We implemented the framework in \Cpp, and tested it on
scenarios involving two-dimensional objects. Nearest-neighbor search,
which is used for the construction of RGGs, was implemented using
\textsc{flann}~\cite{flann}. We note that other nearest-neighbor
search data structures that are tailored for the implementation of
RGGs exist (see, e.g.,~\cite{KleETAL15}). Geometric objects, such as
points, curves, and polygons were represented with
\textsc{cgal}~\cite{cgal}. For the representation of graphs and
related algorithms we used \textsc{boost}~\cite{boost}. Experiments
were conducted on a PC with Intel i7-2600 3.4GHz processor with 8GB of
memory, running a 64-bit Windows~7 OS.

We proceed to describe the implementation involving the computation of
non-trivial cost maps. For curve embedding we used
\textsc{pqp}~\cite{pqp} for collision detection, i.e., determining
whether a given point lies in the forbidden region
$[0,1]^2\setminus \F$. Finally, the cost of an edge with respect to a
given cost map was approximated by dense sampling along the edge, as
is customary in motion planning (see, e.g.,~\cite{LaValle06}).

The majority of running time (over \%90) in the experiments below is
devoted to the computation of $\M$ for given point samples or
edges. Thus, we report only the overall running time in the following
experiments.  We mention that we also implemented a simple grid-based
method for the purpose of comparison with the framework. However, it
performed poorly in easy scenarios and did not terminate in hard
cases. Thus, we chose to omit theses results here.

Unless stated otherwise, we use in the experiments the connection
radius which is described in Theorem~\ref{thm:main_weak}, and denote
it by $r^*_n$. This applies both to the standard and strong regimes of
the problem. A discussion regarding the connection radius in the
strong regime appears below in Section~\ref{sec:exp:strong}.

\subsection{Various scenarios}\label{sec:exp:var}
In this set of experiments we demonstrate the flexibility of the
framework and test it on the three different scenarios. We emphasize
that we employ a shared code framework to solve these three problems
and the ones described later. The only difference in the
implementation lies in the type of  cost function used. The following
problems are solved using a planner for the \emph{strong} case of BPP. 

Figure~\ref{fig:two_curve} (left) depicts an instance of \textbf{P1} (see
Section~\ref{sec:prel:examples}), which
consists of two geometrically-identical curves (red and blue). The
curves are bounded in $[0,1]^2$ and the red curve is translated by
$(0.05,0.05)$ from the blue curve. The optimal solution has a cost of
$0.07$, in which the curves are traversed identically. Our program was
able to produce a solution of cost $0.126$ in $27$ seconds and
$n=\num{100000}$ samples. Results reported throughout this section are the
averaged over $10$ trials.

Figure~\ref{fig:obs} (left) depicts an instance of \textbf{P2}. The
goal is to find a traversal of the three curves such that the
traversal point along the purple curve is visible from either the blue
or red curve, while of course minimizing the lengths of the leashes
between the three curves. Note that the view can be obstructed by the
gray rectangular obstacles. A trivial, albeit poor, solution is to
move the point along the purple curve from start to end, while the
traversal point of, say, the red curve stays put in the start
position. A much better solution, which maintains short leashes, is
described as follows: we move along the purple curve until reaching
the first resting point, indicated by the leftmost black disc. Then we
move along the red curve until we reach to the position directly below
the black circle. Only then we move along the blue curve from start
until reaching the point directly below the first black disc. We use a
similar parametrization with respect to the second ``pit stop'', and
so on. Such a solution was obtained by our program in $11$ seconds
using $n=\num{20000}$ samples.

Figure~\ref{fig:obs} (right) depicts an instance of \textbf{P3}. The input
consists of a curve (depicted in red), and polygonal obstacles
(depicted in gray). The solution obtained by our program after $600$
seconds with $n=\num{100000}$, is drawn in blue. 

\begin{figure}[t]\centering 
  \begin{subfigure}{0.48\textwidth}\centering
    \includegraphics[height=0.8\textwidth]{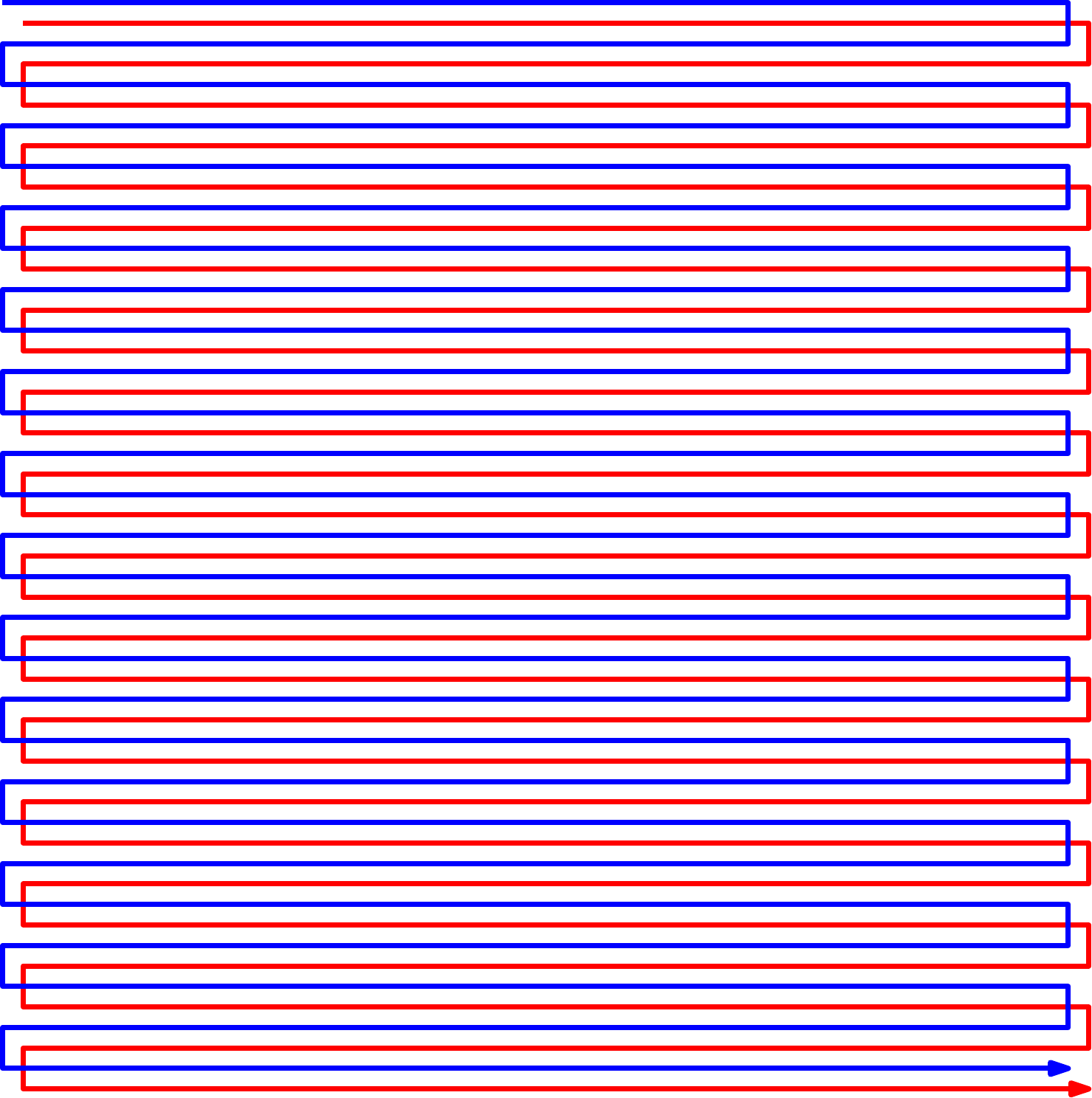}
  \end{subfigure}
  ~
  \begin{subfigure}{0.48\textwidth}\centering
    \includegraphics[width=0.9\textwidth]{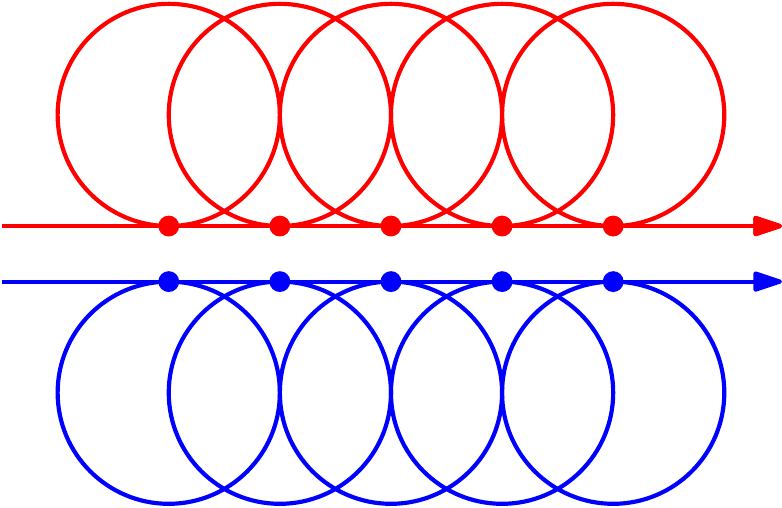}
  \end{subfigure}
  \caption{Scenarios involving two curve.\label{fig:two_curve}}
\end{figure}

\begin{figure}[t] \centering 
  \begin{subfigure}{0.48\textwidth}\centering
    \includegraphics[height=0.8\textwidth]{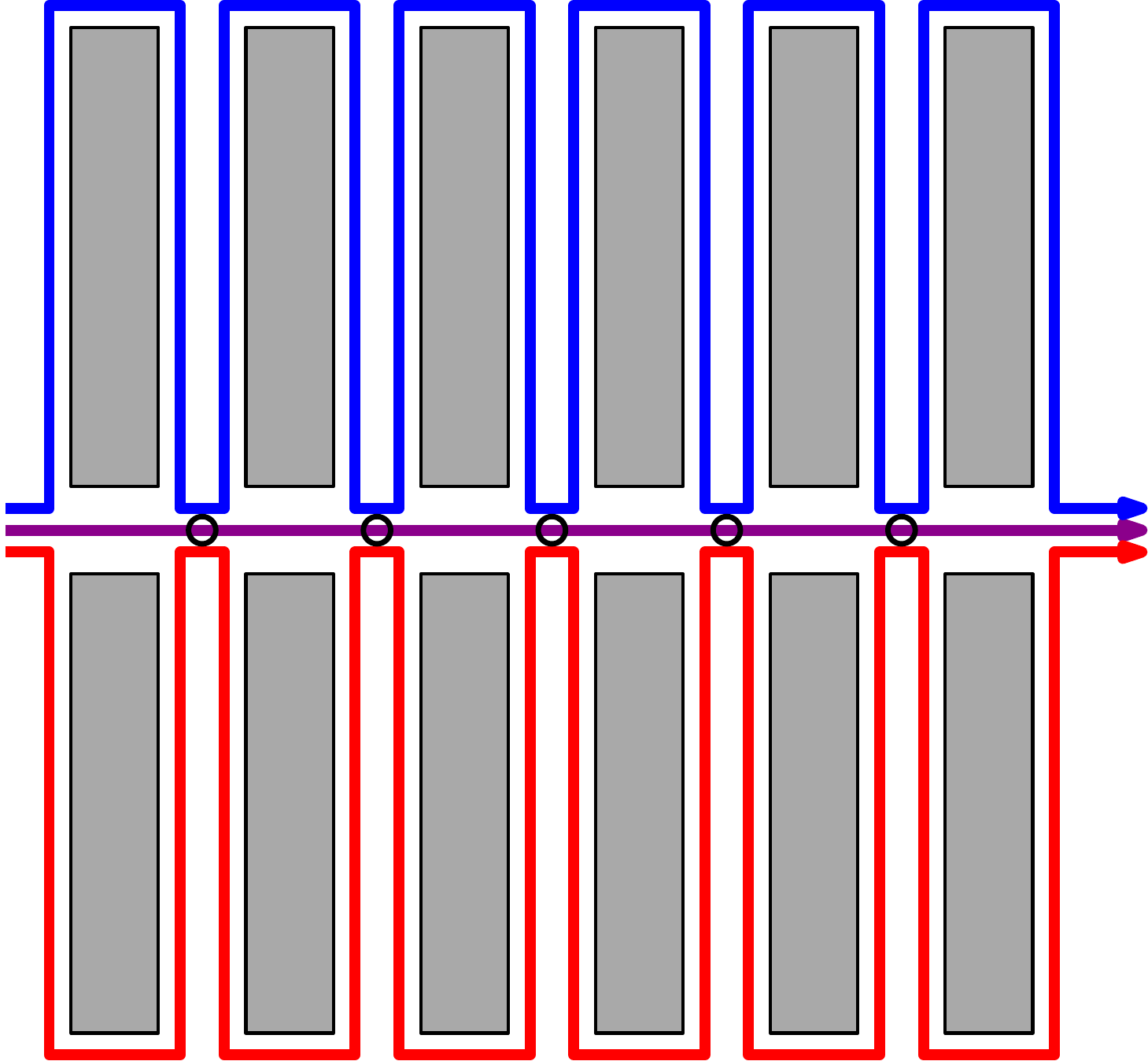}
  \end{subfigure}
  ~
  \begin{subfigure}{0.48\textwidth}\centering
    \includegraphics[width=0.8\textwidth]{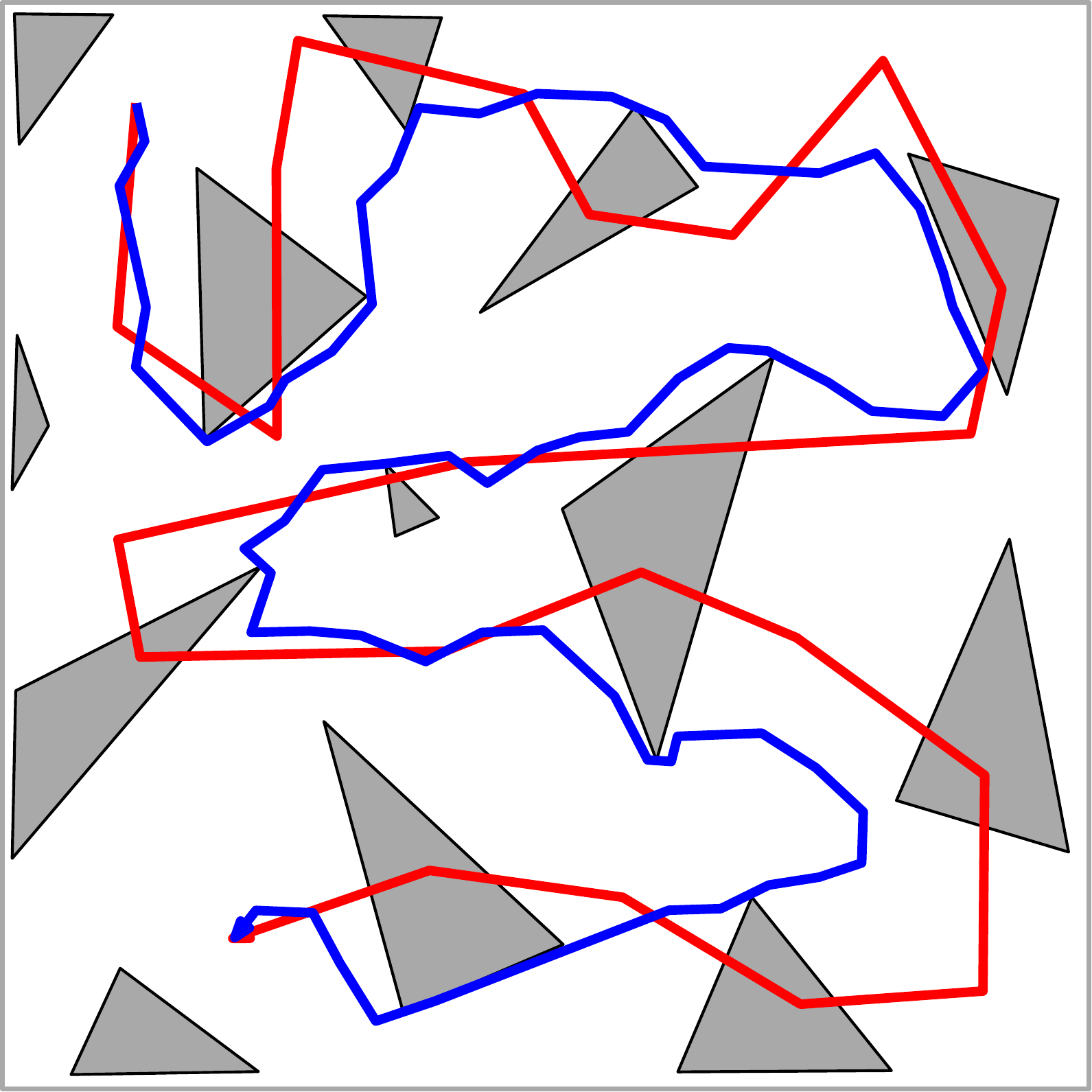}
  \end{subfigure}
  \caption{Scenarios involving curves and obstacles.\label{fig:obs}}
\end{figure}

\subsection{Increasing difficulty}\label{sec:exp:dif}
Here we focus on \textbf{P1} for two curves in the standard regime. We
study how the difficulty of the problem affects the running time and
the rate of convergence of the returned cost. We start with a base
scenario, depicted in Figure~\ref{fig:two_curve} (right), and
gradually increase its difficulty. In the depicted scenario the bottom
(blue) curve consists of five circular loops of radius $0.15$, where
the entrance and exit point to each circle is indicated by a
bullet. The top curve is similarly defined, and the two curves are
separated by a vertical distance of $0.04$. The optimal matching of
cost $0.34$ is obtained in the following manner: when a given circle
of the red curve is traversed, the position along the blue curve is
fixed to the entrance point of the circle directly below the traversed
circle, and vice versa. In a similar fashion we construct scenarios
with 10,20,40 and 80 loops in each
curve. 

In Figure~\ref{graph:weak} we report for each of the scenarios the
cost of the obtained solution as a function of the number of samples
$n$. We set $n=2^i$ for the integer value $i$ between $12$ and
$18$. For $i=12$ and $i=18$ the running times were roughly $2$ and
$66$ seconds, respectively. In between, the values were linearly
proportional to the number of samples (results omitted).  Observe that
as the difficulty of the problem increases the convergence rate of the
cost slightly decreases, but overall a value near the optimum is
reached fairly quickly.

\begin{figure}[t]
  \begin{center}
    \includegraphics[width=0.48\textwidth, trim=0pt 0pt 30pt 20pt ,
    clip=true]{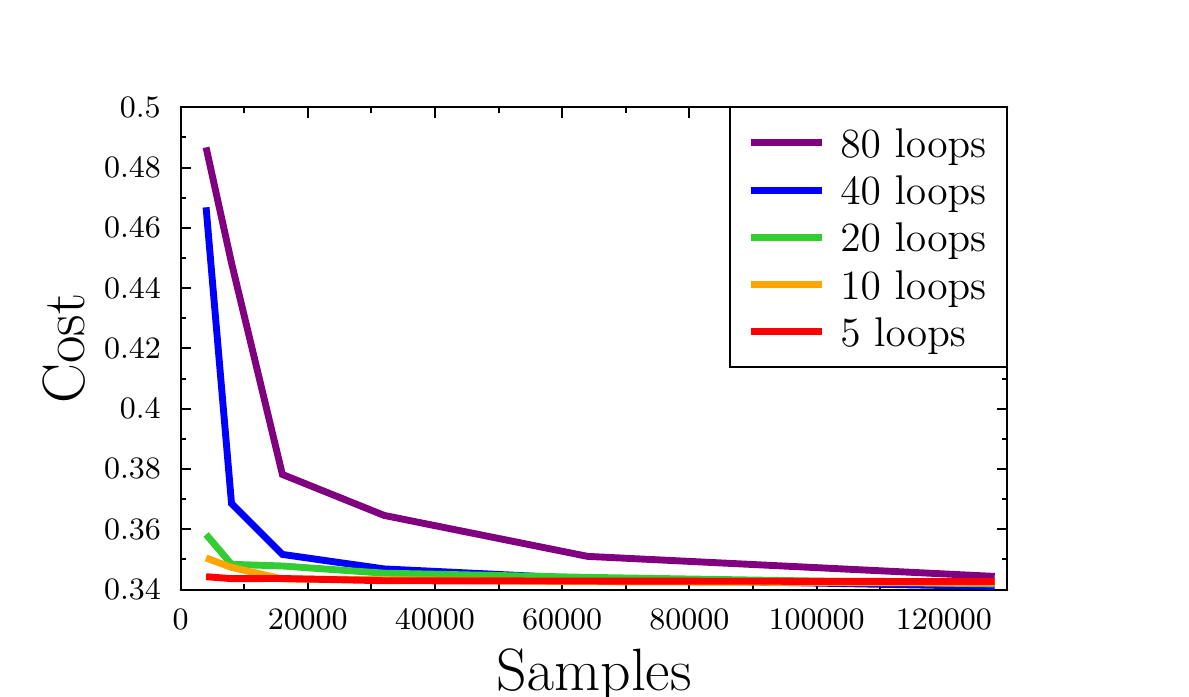}
  \end{center}
  \caption{Results for scenarios of increasing difficulty, as
    described in Section~\ref{sec:exp:dif}.\label{graph:weak}}
\end{figure}

\subsection{Connection radius in the strong
  regime}\label{sec:exp:strong}
Here we consider the \emph{strong} regime and study the behavior of the
framework for varying connection radii. For this purpose, we use the
two-curves scenario with 20 loops that was described in
Section~\ref{sec:exp:dif}. We set the connection radius to
$r_n:=g_n\cdot r_n^*$, where $r_n^*$ is the radius of the standard
regime (see Theorem~\ref{thm:main_weak}). We set
$g_n\in \{1,1.1,\log\log n + 1,\sqrt{\log n}\}$. Results are depicted
in Figure~\ref{graph:strong}.

Not surprisingly, larger values of $r_n$ lead to quicker convergence,
in terms of the number of samples required, to the optimum. However,
this comes at the price of a denser RGG, which results in poor running
times. Note that the program terminated due to lack of space for the
two largest functions of $g_n$ for $n=\num{128000}$. Interestingly, the
connection radius $r^*_n$ of the standard regime seems to converge to
the optimum, albeit slowly. This leads to the question whether such a
function also results in connectivity in the strong regime. Note that
our proof of the convergence in the strong regime requires a larger
value of $r_n$ (see Theorem~\ref{thm:main_strong}).

\begin{figure}[t]\centering 
  \begin{subfigure}{0.48\textwidth}\centering
    \includegraphics[width=1\textwidth, trim=0pt 0pt 30pt 20pt , clip=true]{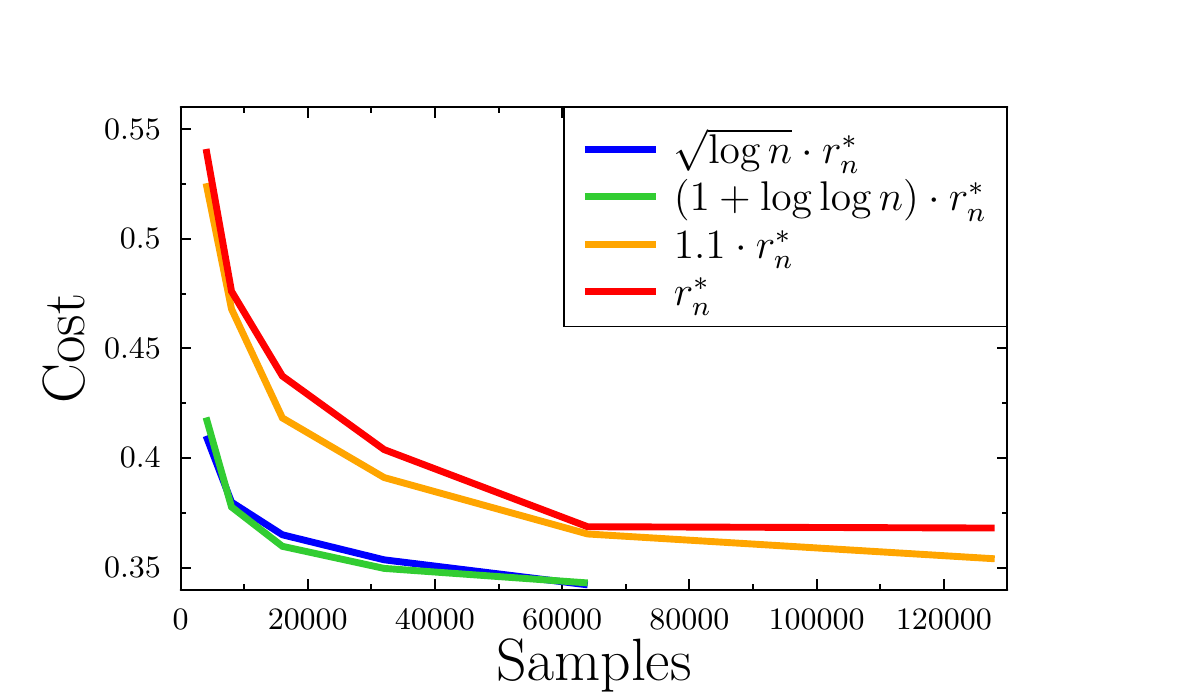}
  \end{subfigure}
  ~
  \begin{subfigure}{0.48\textwidth}\centering
    \includegraphics[width=1\textwidth, trim=0pt 0pt 30pt 20pt , clip=true]{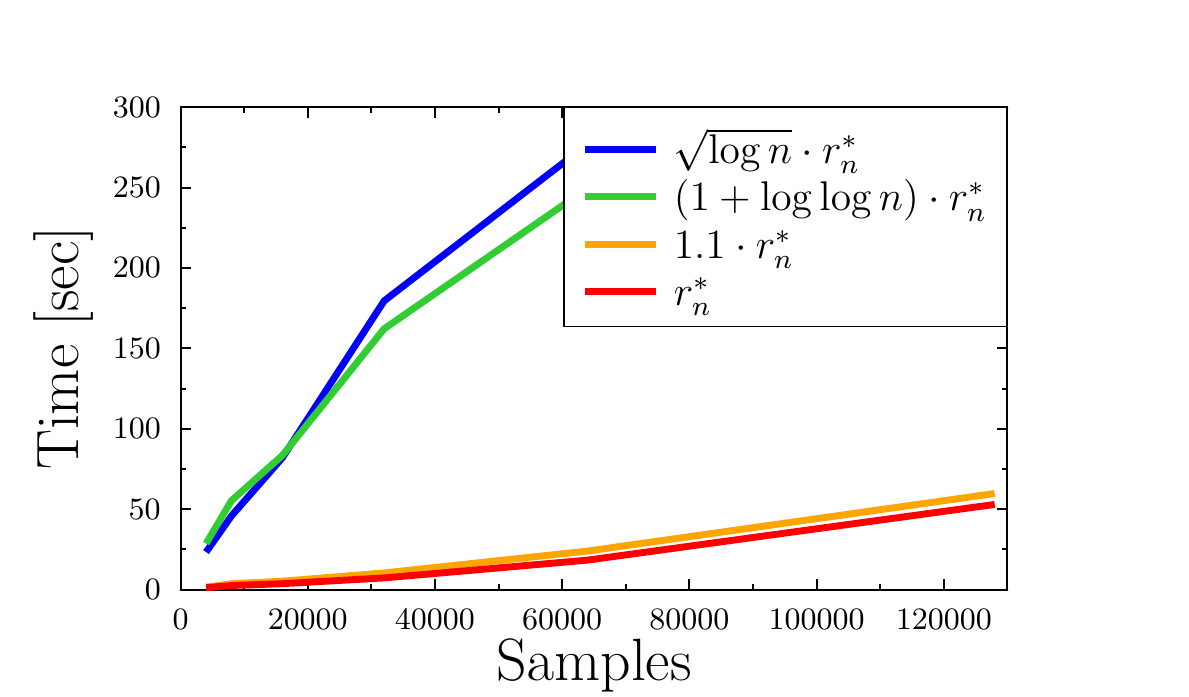}
  \end{subfigure}
  \caption{Results for varying connection radii in the strong regime, as described in
    Section~\ref{sec:exp:strong}.\label{graph:strong}}
\end{figure}

\subsection{Increasing dimensionality}\label{sec:exp:dim}
We test how the dimension of the configuration space $d$ affects
the performance. For this purpose we study the behavior of the
framework on \emph{weak} $k$-curve \frechet distance with $k$ ranging
from~$2$ to~$5$. For~$k=2$ we use the scenario described in
Section~\ref{sec:exp:dif} with~$10$ loops. For~$k=3$ we add another
copy of the blue curve, for $k=4$ an additional copy of the red curve,
and another blue curve for $k=5$. We report running time and cost in
Figure~\ref{graph:k_curve} for various values of $n$, as described
earlier.

Note that that for $k=4$ the program ran out of memory for $n=\num{64000}$,
and for $k=5$ around $n=\num{32000}$. This phenomena occurs since the
connection radius obtained in Theorem~\ref{thm:main_weak} grows
exponentially in $d$.  In particular, for
$r_n=\gamma\left(\frac{\log n}{n}\right)^{1/d}$, where
$\gamma=2(2d\theta_d)^{-1/d}$, each sample has in expectancy
$\Theta(2^d\log n)$ neighbors in the obtained RGG.

\begin{figure} \centering 
  \begin{subfigure}{0.48\textwidth}\centering
    \includegraphics[width=1\textwidth, trim=0pt 0pt 30pt 20pt , clip=true]{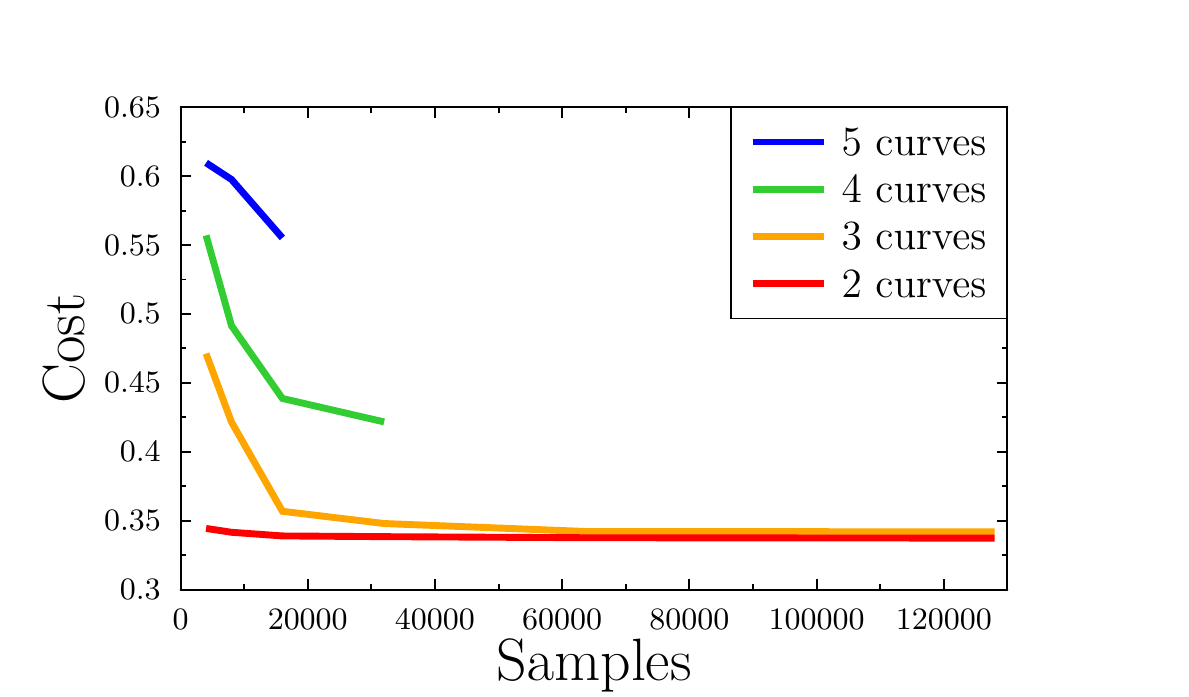}
  \end{subfigure}
  ~
  \begin{subfigure}{0.48\textwidth}\centering
    \includegraphics[width=1\textwidth, trim=0pt 0pt 30pt 20pt , clip=true]{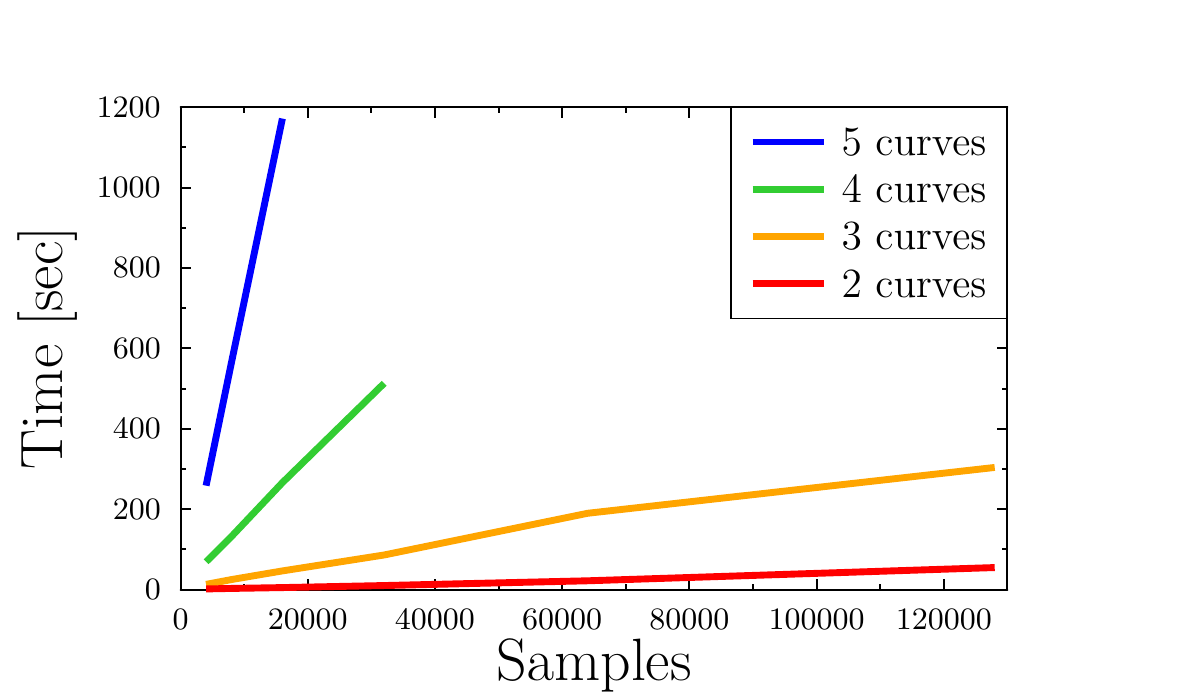}
  \end{subfigure}
  \caption{Figures for the first set of experiments, as described in
    Section~\ref{sec:exp:dim}\label{graph:k_curve}}
\end{figure}

\bibliographystyle{abbrv}
\bibliography{bibliography}

\end{document}